\documentclass[reqno]{amsart}

\usepackage{amsmath}
\usepackage{amsthm}
\usepackage{amsfonts}
\usepackage{amssymb}
\usepackage{graphicx}
\usepackage{color}
\usepackage{ucs}
\usepackage[utf8x]{inputenc}

\definecolor{myurlcolor}{rgb}{0,0,0.7}
\usepackage{hyperref}
\hypersetup{colorlinks,
linkcolor=myurlcolor,
citecolor=myurlcolor,
urlcolor=myurlcolor}

\newcommand{\gt}{>}

\newcommand{\maps}{\colon}

\newtheorem{theorem}{Theorem}

\newtheorem{corol}[theorem]{Corollary}
\newtheorem{defn}[theorem]{Definition}

\newcommand{\FinProb}{\mathtt{FinProb}}
\newcommand{\FinMeas}{\mathtt{FinMeas}}

\title{A Characterization of Entropy in Terms of Information Loss} 

\author{John C. Baez}
\address{Department of Mathematics\\ 
University of California\\ 
Riverside CA 92521\\
and Centre for Quantum Technologies\\ 
National University of Singapore\\ 
Singapore 117543}
\email{baez@math.ucr.edu}

\author{Tobias Fritz}
\address{ICFO -- Institut de Ci\`{e}nces Fot\`{o}niques\\ 
Mediterranean Technology Park\\ 
08860 Castelldefels (Barcelona)\\ 
Spain}
\email{tobias.fritz@icfo.es}

\author{Tom Leinster} 
\address{School of Mathematics and Statistics\\ 
University of Glasgow\\ 
Glasgow G12 8QW\\ 
UK\\ 
and Boyd Orr Centre for Population and Ecosystem Health\\ 
University of Glasgow.}
\email{tom.leinster@glasgow.ac.uk}



\thanks{We thank the denizens of the $n$-Category Caf\'e, especially David
Corfield, Steve Lack, Mark Meckes and Josh Shadlen, for encouragement and
helpful suggestions. Tobias Fritz is supported by the EU STREP QCS. Tom
Leinster is supported by an EPSRC Advanced Research Fellowship.}

\begin{document}

\maketitle

\begin{abstract}
There are numerous characterizations of Shannon entropy and Tsallis entropy as
measures of information obeying certain properties. Using work by Faddeev and
Furuichi, we derive a very simple characterization. Instead of focusing on the
entropy of a probability measure on a finite set, this characterization
focuses on the `information loss', or change in entropy, associated with a
measure-preserving function.    Information loss is a special case of conditional entropy: namely, it is the entropy of a random variable conditioned on some function of that variable.  We show that Shannon entropy gives the only
concept of information loss that is functorial, convex-linear and continuous.
This characterization naturally generalizes to Tsallis entropy as well.
\end{abstract}

\hypertarget{introduction}{}%
\section{{Introduction}}
\label{introduction}

The Shannon entropy~\cite{S} of a probability measure $p$ on a finite set $X$
is given by:
\[
H(p) = -\sum_{i \in X} p_i \ln(p_i).
\]
There are many theorems that seek to characterize Shannon entropy starting
from plausible assumptions; see for example the book by Acz\'el and
Dar\'oczy~\cite{AD}. Here we give a new and very simple characterization
theorem. The main novelty is that we do not focus directly on the entropy of a
single probability measure, but rather, on the \emph{change} in entropy
associated with a measure-preserving function.  The entropy of a single
probability measure can be recovered as the change in entropy of the unique
measure-preserving function onto the one-point space.

A measure-preserving function can map several points to the same point, but
not vice versa, so this change in entropy is always a \emph{decrease}. Since
the second law of thermodynamics speaks of entropy \emph{increase}, this may
seem counterintuitive. It may seem less so if we think of the function as some
kind of data processing that does not introduce any additional randomness. 
Then the entropy can only decrease, and we can talk about the `information
loss' associated with the function.

Some examples may help to clarify this point. Consider the only
possible map $f\maps \{a,b\} \to \{c\}$. Suppose $p$ is the
probability measure on $\{a,b\}$ such that each point has measure
$1/2$, while $q$ is the unique probability measure on the set
$\{c\}$. Then $H(p) = \ln 2$, while $H(q) = 0$. The information loss
associated with the map $f$ is defined to be $H(p) - H(q)$, which in
this case equals $\ln 2$. In other words, the measure-preserving map
$f$ loses one bit of information.

On the other hand, $f$ is also measure-preserving if we replace $p$ by
the probability measure $p'$ for which $a$ has measure $1$ and $b$ has
measure $0$.  Since $H(p') = 0$, the function $f$ now has information
loss $H(p') - H(q) = 0$. It may seem odd to say that $f$ loses no
information: after all, it maps $a$ and $b$ to the the same point.  However, 
because the point $b$ has probability zero with respect to $p'$, knowing 
that $f(x) = c$ lets us conclude that $x = a$ with probability one.

The shift in emphasis from probability measures to
measure-preserving \emph{functions} suggests that it will be useful to
adopt the perspective of category theory~\cite{M}, where one has objects
and \emph{morphisms} between them.  However, the reader
need only know the definition of `category' to understand this paper.

Our main result is that Shannon entropy has a very simple
characterization in terms of information loss. To state it, we
consider a category where a morphism $f\maps p \to q$ is a
measure-preserving function between finite sets equipped with
probability measures. We assume $F$ is a function that assigns to any
such morphism a number $F(f) \in [0,\infty)$, which we call its
\textbf{information loss}. We also assume that $F$ obeys three
axioms. If we call a morphism a `process' (to be thought of as deterministic),
we can state these roughly 
in words as follows.  For the precise statement, including all the
definitions, see Section~\ref{main_result}.
 
\begin{enumerate} \item \textbf{Functoriality}. Given a process consisting of
two stages, the amount of information lost in the whole process is the sum of
the amounts lost at each stage: 
\[ F(f \circ g) = F(f) + F(g) .  \]

\item \textbf{Convex linearity}. If we flip a probability-$\lambda$ coin to
decide whether to do one process or another, the information lost is $\lambda$
times the information lost by the first process plus $(1 - \lambda)$ times the
information lost by the second:
\[
F(\lambda f \oplus (1 - \lambda) g) = \lambda F(f) + (1 - \lambda) F(g).
\]

\item \textbf{Continuity}. If we change a process slightly, the information
lost changes only slightly: $F(f)$ is a continuous function of $f$.
\end{enumerate}

\noindent Given these assumptions, we conclude that there exists a constant $c
\ge 0$ such that for any $f\maps p \to q$, we have
\[
F(f) = c(H(p) - H(q)) .
\]
The charm of this result is that the first two hypotheses look like 
linear conditions, and none of the hypotheses hint at any special role 
for the function $- p \ln p$, but it emerges in the conclusion. The 
key here is a result of Faddeev~\cite{Fa} described in Section~\ref{faddeev}. 

For many scientific purposes, \emph{probability} measures are not enough. Our
result extends to general measures on finite sets, as follows. Any measure on
a finite set can be expressed as $\lambda p$ for some scalar $\lambda$ and
probability measure $p$, and we define $H(\lambda p) = \lambda H(p)$. In this
more general setting, we are no longer confined to taking \emph{convex} linear
combinations of measures. Accordingly, the convex linearity condition in our
main theorem is replaced by two conditions: additivity ($F(f \oplus g) = F(f)
+ F(g)$) and homogeneity ($F(\lambda f) = \lambda F(f)$). As before, the
conclusion is that, up to a multiplicative constant, $F$ assigns to each
morphism $f\maps p \to q$ the information loss $H(p) - H(q)$.

It is natural to wonder what happens when we replace the homogeneity axiom
$F(\lambda f) = \lambda F(f)$ by a more general homogeneity condition: 
\[
F(\lambda f) = \lambda^\alpha F(f)
\]
for some number $\alpha \gt 0$. In this case we find that $F(f)$ is
proportional to $H_\alpha(p) - H_\alpha(q)$, where $H_\alpha$ is the so-called
Tsallis entropy of order $\alpha$.

\hypertarget{main_result}{}%
\section{{The main result}}
\label{main_result}

We work with finite sets equipped with probability measures. All measures on a
finite set $X$ will be assumed nonnegative and defined on the $\sigma$-algebra
of all subsets of $X$.  Any such measure is determined by its values on
singletons, so we will think of a probability measure $p$ on $X$ as an
$X$-tuple of numbers $p_i \in [0,1]$ ($i \in X$) satisfying $\sum p_i = 1$.

\begin{defn}
\label{FinProb}\hypertarget{FinProb}{}
Let $\FinProb$ be the category where an object $(X,p)$ is given by a finite set $X$ equipped with
a probability measure $p$, and where a morphism $f\maps (X, p)\to (Y, q)$ 
is a measure-preserving function from $(X,p)$ to $(Y,q)$, that is, a function 
$f\maps X\to Y$ such that
\[
q_j = \sum_{i \in f^{-1}(j)} p_i 
\]
for all $j \in Y$.
\end{defn}

\noindent
We will usually write an object $(X,p)$ as $p$ for short,
and write a morphism $f \maps (X,p) \to (Y,q)$ as simply $f: p \to q$.

There is a way to take convex linear combinations of objects and
morphisms in $\FinProb$. Let $(X, p)$ and $(Y, q)$ be finite sets equipped
with probability measures, and let $\lambda \in [0, 1]$. Then there is a
probability measure
\[
\lambda p \oplus (1 - \lambda) q
\]
on the disjoint union of the sets $X$ and $Y$, whose value at a point $k$ is
given by
\[
(\lambda p \oplus (1 - \lambda) q)_k
=
\begin{cases}
\lambda p_k       &\text{if } k \in X\\
(1 - \lambda) q_k &\text{if } k \in Y.
\end{cases}
\]
Given morphisms $f: p \to p'$ and $g: q \to q'$, there is a unique morphism
\[
\lambda f \oplus (1 - \lambda) g:
\lambda p \oplus (1 - \lambda) q \longrightarrow
\lambda p' \oplus (1 - \lambda) q'
\]
that restricts to $f$ on the measure space $p$ and to $g$ on the measure space
$q$. 

The same notation can be extended, in the obvious way, to convex combinations
of more than two objects or morphisms. For example, given objects $p(1),
\ldots, p(n)$ of $\FinProb$ and nonnegative scalars $\lambda_1, \ldots,
\lambda_n$ summing to $1$, there is a new object $\bigoplus_{i = 1}^n
\lambda_i p(i)$.

Recall that the \textbf{Shannon entropy} of a probability measure $p$ on a
finite set $X$ is
\[
H(p) = -\sum_{i \in X} p_i \ln(p_i) \in [0, \infty),
\]
with the convention that $0 \ln(0) = 0$.

\begin{theorem}
\label{characterization_prob}\hypertarget{characterization_prob}{}
Suppose $F$ is any map sending morphisms in $\FinProb$ to numbers in
$[0,\infty)$ and obeying these three axioms:
\begin{enumerate}%
\item \textbf{Functoriality}:
\label{ch_prob_functor}
\begin{equation}
F(f \circ g) = F(f) + F(g)
\label{functoriality}
\end{equation}
whenever $f,g$ are composable morphisms.

\item \textbf{Convex linearity}:
\begin{equation}
F(\lambda f \oplus (1 - \lambda)g) 
= \lambda F(f) + (1 - \lambda) F(g)
\label{convex_linearity}
\end{equation}
for all morphisms $f,g$ and scalars $\lambda \in [0, 1]$.

\item \textbf{Continuity}: $F$ is continuous.
\label{ch_prob_cont}
\end{enumerate}
Then there exists a constant $c \ge 0$ such that for any morphism $f : p \to
q$ in $\FinProb$,
\[
F(f) = c(H(p) - H(q))
\]
where $H(p)$ is the Shannon entropy of $p$. Conversely, for any constant $c
\ge 0$, this formula determines a map $F$ obeying conditions
\ref{ch_prob_functor}--\ref{ch_prob_cont}.
\end{theorem}

We need to explain condition~\ref{ch_prob_cont}. A sequence of morphisms
\[
(X_n, p(n)) \stackrel{f_n}{\to} (Y_n, q(n))
\]
in $\FinProb$ \textbf{converges} to a morphism $(X, p) \stackrel{f}{\to} (Y,
q)$ if:
\begin{itemize}%
\item for all sufficiently large $n$, we have $X_n = X$, $Y_n = Y$, and
$f_n(i) = f(i)$ for all $i \in X$;

\item $p(n) \to p$ and $q(n) \to q$ pointwise.
\end{itemize}
We define $F$ to be \textbf{continuous} if $F(f_n) \to F(f)$ whenever $f_n$ is
a sequence of morphisms converging to a morphism $f$. 

The proof of Theorem~\ref{characterization_prob} is given in
Section~\ref{proof}.  First we show how to deduce a characterization of
Shannon entropy for general measures on finite sets.

The following definition is in analogy to Definition~\ref{FinProb}:

\begin{defn}
\label{FinMeas}\hypertarget{FinMeas}{}
Let $\FinMeas$ be the category whose objects are finite sets equipped with
measures and whose morphisms are measure-preserving functions. 
\end{defn}

There is more room for maneuver in $\FinMeas$ than in $\FinProb$: we can take
arbitrary nonnegative linear combinations of objects and morphisms, not just
convex combinations. Any nonnegative linear combination can be built up from
direct sums and multiplication by nonnegative scalars, which are defined as
follows. 

\begin{itemize}%
\item For direct sums, first note that the disjoint union of two finite sets
equipped with measures is another object of the same type. We write the
disjoint union of $p, q \in \FinMeas$ as $p \oplus q$. Then, given morphisms
$f \maps p \to p'$, $g \maps q \to q'$ there is a unique morphism 
$f \oplus g \maps p \oplus q \to p' \oplus q'$ that restricts to $f$ on 
the measure space $p$ and to $g$ on the measure space $q$.

\item For scalar multiplication, first note that we can multiply a
measure by a nonnegative real number and get a new measure. So, given
an object $p \in \FinMeas$ and a number $\lambda \ge 0$ we obtain an
object $\lambda p \in \FinMeas$ with the same underlying set and with
$(\lambda p)_i = \lambda p_i$.  Then, given a morphism $f \maps p \to
q$, there is a unique morphism $\lambda f \maps \lambda p \to \lambda q$
that has the same underlying function as $f$.
\end{itemize}

\noindent
This is consistent with our earlier notation for convex linear combinations.

We wish to give some conditions guaranteeing that a map sending morphisms in
$\FinMeas$ to nonnegative real numbers comes from a multiple of Shannon
entropy. To do this we need to define the Shannon entropy of a finite set $X$
equipped with a measure $p$, not necessarily a probability measure. Define the
\textbf{total mass} of $(X, p)$ to be
\[
\|p\| = \sum_{i \in X} p_i .
\]
If this is nonzero, then $p$ is of the form $\|p\| \bar{p}$ for a unique
probability measure space $\bar{p}$. In that case we define the
\textbf{Shannon entropy} of $p$ to be $\|p\| H(\bar{p})$. If the total mass of
$p$ is zero, we define its Shannon entropy to be zero.

We can define continuity for a map sending morphisms in $\FinMeas$ to numbers
in $[0,\infty)$ just as we did for $\FinProb$, and show:

\begin{corol}
\label{characterization_meas}\hypertarget{characterization_meas}{}
Suppose $F$ is any map sending morphisms in $\FinMeas$ to numbers in
$[0,\infty)$ and obeying these four axioms: 
\begin{enumerate}%
\item \textbf{Functoriality}:
\label{ch_meas_functor}
\[
F(f \circ g) = F(f) + F(g)
\]
whenever $f,g$ are composable morphisms.

\item \textbf{Additivity}:
\begin{equation}
F(f \oplus g) = F(f) + F(g)
\label{additivity}
\end{equation}
for all morphisms $f,g$.

\item \textbf{Homogeneity}:
\begin{equation}
F(\lambda f) = \lambda F(f)
\label{multiplicativity}
\end{equation}
for all morphisms $f$ and all $\lambda \in [0,\infty)$.

\item Continuity: $F$ is continuous.
\label{ch_meas_cont}
\end{enumerate}
Then there exists a constant $c \ge 0$ such that for any morphism 
$f \maps p \to
q$ in $\FinMeas$, 
\[
F(f) = c(H(p) - H(q))
\]
where $H(p)$ is the Shannon entropy of $p$. Conversely, for any constant $c
\ge 0$, this formula determines a map $F$ obeying conditions
\ref{ch_meas_functor}--\ref{ch_meas_cont}. 
\end{corol}

\begin{proof}
Take a map $F$ obeying these axioms. Then $F$ restricts to a map on
morphisms of $\FinProb$ obeying the axioms of
Theorem~\ref{characterization_prob}. Hence there exists a constant $c
\ge 0$ such that $F(f) = c(H(p) - H(q))$ whenever $f\maps p \to q$ is
a morphism between probability measures. Now take an arbitrary
morphism $f\maps p \to q$ in $\FinMeas$. Since $f$ is measure-preserving,
$\|p\| = \|q\| = \lambda$, say.  If $\lambda \neq 0$ then $p = \lambda
\bar{p}$, $q = \lambda \bar{q}$ and $f = \lambda \bar{f}$ for some
morphism $\bar{f}\maps \bar{p} \to \bar{q}$ in $\FinProb$; then by
homogeneity,
\[
F(f) 
= \lambda F(\bar{f}) 
= \lambda c (H(\bar{p}) - H(\bar{q})) 
= c(H(p) - H(q)).
\]
If $\lambda = 0$ then $f = 0 f$, so $F(f) = 0$ by homogeneity. So $F(f) =
c(H(p) - H(q))$ in either case. The converse statement follows from the
converse in Theorem~\ref{characterization_prob}. 
\end{proof}

\hypertarget{why_shannon}{}%
\section{{Why Shannon entropy works}}
\label{why_shannon}

To prove the easy half of Theorem~\ref{characterization_prob}, we must check
that $F(f) = c(H(p) - H(q))$ really does determine a functor obeying all the
conditions of that theorem. Since all these conditions are linear in $F$, it
suffices to consider the case where $c = 1$. It is clear that $F$ is
continuous, and equation~\eqref{functoriality} is also immediate whenever 
$g\maps m \to p$, $f\maps p \to q$, are morphisms in $\FinProb$:
\[
F(f \circ g) = H(m) - H(q) = H(p) - H(q) + H(m) - H(p) = F(f) + F(g).
\]
The work is to prove equation~\eqref{convex_linearity}.

We begin by establishing a useful formula for $F(f) = H(p) - H(q)$, where as
usual $f$ is a morphism $p \to q$ in $\FinProb$. Since $f$ is
measure-preserving, we have
\[
q_j = \sum_{i \in f^{-1}(j)} p_i.
\]
So
\begin{align*}
\sum_j q_j \ln q_j &
= \sum_j \sum_{i \in f^{-1}(j)} p_i \ln q_j \\
&
= \sum_j \sum_{i \in f^{-1}(j)} p_i \ln q_{f(i)} \\
&
= \sum_i p_i \ln q_{f(i)} 
\end{align*}
where in the last step we note that summing over all $i$ that map to
$j$ and then summing over all $j$ is the same as summing over all $i$. So,
\begin{align*}
F(f) &
= - \sum_i p_i\ln p_i + \sum_j q_j \ln q_j   \\
&
= \sum_i ( -p_i \ln p_i + p_i \ln q_{f(i)}) 
\end{align*}
and thus
\begin{equation}
F(f) = \sum_{i \in X}  p_i \ln \frac{q_{f(i)}}{p_i}
\label{entropy_of_map}
\end{equation}
where the quantity in the sum is defined to be zero when $p_i = 0$. If we
think of $p$ and $q$ as the distributions of random variables $x \in X$ and $y
\in Y$ with $y = f(x)$, then $F(f)$ is exactly the conditional entropy of $x$
given $y$. So, what we are calling `information loss' is a special case of
conditional entropy.

This formulation makes it easy to check equation~\eqref{convex_linearity},
\[
F (\lambda f \oplus (1 - \lambda)g) = \lambda F(f) + (1 - \lambda) F(g),
\]
simply by applying~(\ref{entropy_of_map}) on both sides.

In the proof of Corollary~\ref{characterization_meas} (on $\FinMeas$), the
fact that $F(f) = c(H(p) - H(q))$ satisfies the four axioms was deduced from
the analogous fact for $\FinProb$. It can also be checked directly. For this
it is helpful to note that
\begin{equation}
H(p) = \|p\| \ln\|p\| - \sum_i p_i \ln(p_i).
\label{entropy_formula}
\end{equation}
It can then be shown that equation~\eqref{entropy_of_map} holds for
\emph{every} morphism $f$ in $\FinMeas$. The additivity and homogeneity axioms
follow easily.

\hypertarget{faddeev}{}%
\section{{Faddeev's theorem}}
\label{faddeev}

To prove the hard part of Theorem~\ref{characterization_prob}, we use a
characterization of entropy given by Faddeev~\cite{Fa} and nicely summarized
at the beginning of a paper by R\'enyi~\cite{R}. In order to state this result,
it is convenient to write a probability measure on the set $\{1, \dots, n\}$
as an $n$-tuple $p = (p_1, \dots, p_n)$. With only mild cosmetic changes,
Faddeev's original result states:

\begin{theorem}
\label{Faddeev}\hypertarget{Faddeev}{}
\textbf{(Faddeev)} Suppose $I$ is a map sending any probability measure on any
finite set to a nonnegative real number. Suppose that: 
\begin{enumerate}%
\item 
\label{fadd_invt}
$I$ is invariant under bijections.

\item 
\label{fadd_cts}
$I$ is continuous.

\item For any probability measure $p$ on a set of the form $\{1, \dots, n\}$,
and any number $0 \le t \le 1$, 
\label{magical_prop}
\begin{equation}
I((t p_1, (1-t)p_1, p_2, \dots, p_n)) =  
I((p_1, \dots, p_n)) + p_1 I((t, 1 - t)).
\label{faddeev_equation}
\end{equation}
\end{enumerate}
Then $I$ is a constant nonnegative multiple of Shannon entropy.
\end{theorem}

In condition~\ref{fadd_invt} we are using the fact that given a
bijection $f\maps X \to X'$ between finite sets and a probability
measure on $X$, there is a unique probability measure on $X'$ such
that $p$ is measure-preserving; we demand that $I$ takes the same
value on both these probability measures. In condition~\ref{fadd_cts},
we use the standard topology on the simplex
\[
\Delta^{n-1} 
= 
\Bigl\{(p_1,\ldots,p_n) \in \mathbb{R}^n \:\Big|\: 
p_i \geq 0, \: \sum_i p_i = 1 \Bigr\}
\]
to put a topology on the set of probability distributions on any $n$-element
set. 

The most interesting condition in Faddeev's theorem is~\ref{magical_prop}. It
is known in the literature as the `grouping rule'~\cite[2.179]{CT} or
`recursivity'~\cite[1.2.8]{AD}.  It is a special case of `strong
additivity' \cite[1.2.6]{AD}, which already appears in the work of
Shannon~\cite{S} and Faddeev~\cite{Fa}.  Namely, suppose that $p$ is a 
probability measure on the set $\{1,\dots,n \}$.  Suppose also that for 
each $i \in \{1, \ldots, n\}$, we have a probability measure $q(i)$ on 
a finite set $X_i$. Then $p_1 q(1) \oplus \cdots \oplus p_n q(n)$ 
is again a probability measure space, and the Shannon entropy of this 
space is given by the \textbf{strong additivity} formula:
\[
H\Bigl(p_1 q(1) \oplus \cdots \oplus p_n q(n)\Bigr) 
= 
H(p) + \sum_{i=1}^n p_i H(q(i)) .
\]
This can easily be verified using the definition of Shannon entropy and
elementary properties of the logarithm. Moreover, condition~\ref{magical_prop}
in Faddeev's theorem is equivalent to strong additivity together with the
condition that $I((1)) = 0$, allowing us to reformulate Faddeev's theorem as
follows:

\begin{theorem}
\label{Faddeev2}\hypertarget{Faddeev2}{}
Suppose $I$ is a map sending any probability measure on any finite set to a
nonnegative real number. Suppose that: 
\begin{enumerate}%
\item $I$ is invariant under bijections.
\label{bij_inv}

\item $I$ is continuous.
\label{I_cont}

\item $I((1)) = 0$, where $(1)$ is our name for the unique probability measure
on the set $\{1\}$. 
\label{dirac_no_entropy}

\item For any probability measure $p$ on the set $\{1,\dots, n\}$ and
probability measures $q(1),\ldots,q(n)$ on finite sets, we have 
\label{operadic_prop}
\[
I(p_1 q(1) \oplus \cdots \oplus p_n q(n)) = I(p) + \sum_{i=1}^n
p_i I(q(i)) .
\]
\end{enumerate}
Then $I$ is a constant nonnegative multiple of Shannon entropy. Conversely,
any constant nonnegative multiple of Shannon entropy satisfies conditions
\ref{bij_inv}--\ref{operadic_prop}. 
\end{theorem}

\begin{proof}
Since we already know that the multiples of Shannon entropy have all 
these properties, we just need to check that conditions~\ref{dirac_no_entropy}
and~\ref{operadic_prop} imply Faddeev's equation~\eqref{faddeev_equation}.
Take $p = (p_1, \ldots, p_n)$, $q(1) = (t, 1 - t)$ and $q(i) = (1)$ for $i \geq
2$: then condition~\ref{operadic_prop} gives
\[
I((t p_1, (1 - t)p_1, p_2, \ldots, p_n))
=
I((p_1, \ldots, p_n)) + p_1 I((t, 1 - t)) + \sum_{i = 2}^n p_i I((1))
\]
which by condition~\ref{dirac_no_entropy} gives Faddeev's equation.
\end{proof}

It may seem miraculous how the formula
\[
I(p_1, \dots, p_n) = - c \sum_i p_i \ln p_i
\]
emerges from the assumptions in either Faddeev's original
Theorem~\ref{Faddeev} or the equivalent Theorem~\ref{Faddeev2}. We can
demystify this by describing a key step in Faddeev's argument, as simplified
by R\'enyi~\cite{R}. Suppose $I$ is a function satisfying the assumptions of
Faddeev's result. Let
\[
\phi(n) = I \biggl( \frac{1}{n} , \dots, \frac{1}{n} \biggr)
\]
equal $I$ applied to the uniform probability measure on an
$n$-element set. Since we can write a set with $n m$ elements as a disjoint
union of $m$ different $n$-element sets, condition~\ref{operadic_prop} of
Theorem~\ref{Faddeev2} implies that
\[
\phi(n m) = \phi(n) + \phi(m).
\]
The conditions of Faddeev's theorem also imply 
\[
\lim_{n \to \infty} (\phi(n+1) - \phi(n)) = 0
\]
and the only solutions of both these equations are 
\[
\phi(n) = c \ln n .
\]
This is how the logarithm function enters. Using condition~\ref{magical_prop}
of Theorem~\ref{Faddeev}, or equivalently conditions~\ref{dirac_no_entropy}
and~\ref{operadic_prop} of Theorem~\ref{Faddeev2}, the value of $I$ can be
deduced for probability measures $p$ such that each $p_i$ is rational. The
result for arbitrary probability measures follows by continuity.

\hypertarget{proof}{}%
\section{{Proof of the main result}}
\label{proof}

Now we complete the proof of Theorem~\ref{characterization_prob}. Assume that
$F$ obeys conditions \ref{ch_prob_functor}--\ref{ch_prob_cont} in the statement
of this theorem.

Recall that $(1)$ denotes the set $\{1\}$ equipped with its unique probability
measure. For each object $p \in \FinProb$, there is a unique morphism
\[
!_p \maps p \to (1).
\]
We can think of this as the map that crushes $p$ down to a point and loses all
the information that $p$ had. So, we define the `entropy' of the measure $p$
by
\[
I(p) = F(!_p) .
\]
Given any morphism $f\maps p \to q$ in $\FinProb$, we have
\[
!_p = !_q \circ f.
\]
So, by our assumption that $F$ is functorial,
\[
F(!_p) = F(!_q) + F(f),
\]
or in other words:
\begin{equation}
F(f) = I(p) - I(q) .
\label{difference}
\end{equation}
To conclude the proof, it suffices to show that $I$ is a multiple of Shannon
entropy.

We do this by using Theorem~\ref{Faddeev2}. Functoriality implies that when a
morphism $f$ is invertible, $F(f) = 0$. Together with~\eqref{difference}, this
gives condition~\ref{bij_inv} of Theorem~\ref{Faddeev2}. Since $!_{(1)}$ is
invertible, it also gives condition~\ref{dirac_no_entropy}.
Condition~\ref{I_cont} is immediate. The real work is checking
condition~\ref{operadic_prop}.

Given a probability measure $p$ on $\{1, \ldots, n\}$ together with
probability measures $q(1), \ldots, q(n)$ on finite sets $X_1, \ldots,
X_n$, respectively, we obtain a probability measure $\bigoplus_i
p_i q(i)$ on the disjoint union of $X_1, \ldots, X_n$.  We can also
decompose $p$ as a direct sum:
\begin{equation}
p\cong \bigoplus_i p_i (1).
\label{decomposition}
\end{equation}
Define a morphism
\[
f = \bigoplus_i p_i !_{q(i)} \maps
\bigoplus_i p_i q(i) 
\to 
\bigoplus_i p_i (1).
\]
Then by convex linearity and the definition of $I$,
\[
F(f) = \sum_i p_i F(!_{q(i)}) 
= \sum_i p_i I(q(i)).
\]
But also
\[
F(f) 
= I\bigl(\bigoplus_i p_i q(i)\bigr) - I\bigl(\bigoplus p_i (1)\bigr)
= I\bigl(\bigoplus_i p_i q(i)\bigr) - I(p)
\]
by~\eqref{difference} and~\eqref{decomposition}. Comparing these two
expressions for $F(f)$ gives condition~\ref{operadic_prop} of
Theorem~\ref{Faddeev2}, which completes the proof of
Theorem~\ref{characterization_prob}.

\hypertarget{a_characterization_of_tsallis_entropy_17}{}%
\section{{A characterization of Tsallis entropy}}
\label{a_characterization_of_tsallis_entropy_17}

Since Shannon defined his entropy in 1948, it has been generalized in many
ways. Our Theorem~\ref{characterization_prob} can easily be extended to
characterize one family of generalizations, the so-called `Tsallis entropies'.
For any positive real number $\alpha$, the \textbf{Tsallis entropy of order
$\alpha$} of a probability measure $p$ on a finite set $X$ is defined as:
\[
H_\alpha(p)
=
\begin{cases}
\displaystyle
\frac{1}{\alpha - 1} 
\biggl(
1 - \sum_{i \in X} p_i^\alpha
\biggr)
& \textrm{if } \alpha \neq 1 \\[0.4cm]
\displaystyle
- \sum_{i \in X} p_i \ln p_i
& \textrm{if } \alpha = 1.
\end{cases}
\]
The peculiarly different definition when $\alpha = 1$ is explained by the fact
that the limit $\lim_{\alpha \to 1} H_\alpha(p)$ exists and equals the Shannon
entropy $H(p)$.

Although these entropies are most often named after Tsallis~\cite{T}, they and
related quantities had been studied by others long before the 1988 paper in
which Tsallis first wrote about them. For example, Havrda and
Charv\'at~\cite{HC} had already introduced a similar formula, adapted to base
$2$ logarithms, in a 1967 paper in information theory, and in 1982, Patil and
Taillie~\cite{PT} had used $H_\alpha$ itself as a measure of biological
diversity.

The characterization of Tsallis entropy is exactly the same as that of Shannon
entropy except in one respect: in the convex linearity condition, the degree
of homogeneity changes from $1$ to $\alpha$.

\begin{theorem}
\label{tsallis_characterization_prob}\hypertarget{tsallis_characterization_prob}{} 
Let $\alpha\in(0,\infty)$. Suppose $F$ is any map sending morphisms in
$\FinProb$ to numbers in $[0,\infty)$ and obeying these three axioms:
\begin{enumerate}%
\item \textbf{Functoriality}:
\label{ts_prob_functor}
\[
F(f \circ g) = F(f) + F(g)
\]
whenever $f,g$ are composable morphisms.

\item \textbf{Compatibility with convex combinations}:
\[
F(\lambda f \oplus (1-\lambda) g) = 
\lambda^{\alpha} F(f) + (1-\lambda)^{\alpha} F(g)
\]
for all morphisms $f,g$ and all $\lambda \in [0,1]$.

\item \textbf{Continuity}: $F$ is continuous.
\label{ts_prob_cont}
\end{enumerate}
Then there exists a constant $c \ge 0$ such that for any morphism 
$f \maps p \to
q$ in $\FinProb$,
\[
F(f) = c(H_{\alpha}(p) - H_{\alpha}(q))
\]
where $H_{\alpha}(p)$ is the order $\alpha$ Tsallis entropy of $p$.
Conversely, for any constant $c \ge 0$, this formula determines a map $F$
obeying conditions \ref{ts_prob_functor}--\ref{ts_prob_cont}.
\end{theorem}

\begin{proof}
We use Theorem~V.2 of Furuichi~\cite{Fu}.  The statement of Furuichi's theorem
is the same as that of Theorem~\ref{Faddeev} (Faddeev's theorem), except that
condition~\ref{magical_prop} is replaced by   
\[
I((t p_1, (1-t)p_1, p_2, \dots, p_n)) 
=  
I((p_1, \dots, p_n)) + p_1^\alpha I((t, 1 - t))
\]
and Shannon entropy is replaced by Tsallis entropy of order $\alpha$.  The
proof of the present theorem is thus the same as that of
Theorem~\ref{characterization_prob}, except that Faddeev's theorem is replaced
by Furuichi's.
%
%
\end{proof}

As in the case of Shannon entropy, this result can be extended to arbitrary
measures on finite sets. For this we need to define the Tsallis entropies of
an arbitrary measure on a finite set. We do so by requiring that
\[
H_\alpha(\lambda p) = \lambda^\alpha H_\alpha(p)
\]
for all $\lambda \in [0, \infty)$ and all $p \in \FinMeas$. When $\alpha = 1$
this is the same as the Shannon entropy, and when $\alpha \neq 1$, we have
\[
H_\alpha(p) =
\frac{1}{\alpha - 1} 
\biggl(
\biggl( \sum_{i \in X} p_i \biggr)^\alpha
-
\sum_{i \in X} p_i^\alpha
\biggr)
\]
(which is analogous to~\eqref{entropy_formula}). The following result is the
same as Corollary~\ref{characterization_meas} except that, again, the degree of
homogeneity changes from $1$ to $\alpha$.

\begin{corol}
\label{tsallis_characterization_meas}%
\hypertarget{tsallis_characterization_meas}{} 
Let $\alpha \in (0, \infty)$. Suppose $F$ is any map sending morphisms in
$\FinMeas$ to numbers in $[0,\infty)$, and obeying these four properties: 
\begin{enumerate}%
\item \textbf{Functoriality}:
\label{ts_meas_functor}
\[
F(f \circ g) = F(f) + F(g)
\]
whenever $f,g$ are composable morphisms.

\item \textbf{Additivity}:
\[
F(f \oplus g) = F(f) + F(g)
\]
for all morphisms $f,g$.

\item \textbf{Homogeneity of degree $\alpha$}:
\[
F(\lambda f) = \lambda^\alpha F(f)
\]
for all morphisms $f$ and all $\lambda \in [0,\infty)$.

\item \textbf{Continuity}: $F$ is continuous.
\label{ts_meas_cont}
\end{enumerate}
Then there exists a constant $c \ge 0$ such that for any morphism $
f\maps p \to q$ in $\FinMeas$, 
\[
F(f) = c(H_\alpha(p) - H_\alpha(q))
\]
where $H_\alpha$ is the Tsallis entropy of order $\alpha$. Conversely, for any
constant $c \ge 0$, this formula determines a map $F$ obeying conditions
\ref{ts_meas_functor}--\ref{ts_meas_cont}. 
\end{corol}

\begin{proof}
This follows from Theorem~\ref{tsallis_characterization_prob} in just the same
way that Corollary~\ref{characterization_meas} follows from
Theorem~\ref{characterization_prob}.
\end{proof}

\section*{Acknowledgements}
We thank the denizens of the $n$-Category Caf\'e, especially David
Corfield, Steve Lack, Mark Meckes and Josh Shadlen, for encouragement and
helpful suggestions. Tobias Fritz is supported by the EU STREP QCS. Tom
Leinster is supported by an EPSRC Advanced Research Fellowship.

\bibliographystyle{mdpi}
\makeatletter
\renewcommand\@biblabel[1]{#1. }
\makeatother

\end{document}